\documentclass[12pt]{article}
\usepackage{latexsym}
\usepackage{amsfonts}
\def\mymedskip{\vskip\medskipamount}
\def\mymedbreak{\par \ifdim\lastskip<\medskipamount
  \removelastskip \penalty-100 \mymedskip \fi}
\def\myaftermedspace{\par \ifdim\lastskip<\medskipamount
  \removelastskip \penalty55\mymedskip\fi}
\newcommand{\eop}{{\unskip\nobreak\hfil\penalty50
          \hskip2em\hbox{}\nobreak\hfil$\Box$
          \parfillskip=0pt \finalhyphendemerits=0 \par}}
\newenvironment{proof}%
{\mymedbreak{\noindent\bf Proof:\enspace}}{\eop\myaftermedspace}
{\mymedbreak{\noindent\bf Proof of Theorem #1:\enspace}}{\eop\myaftermedspace}
{\mymedbreak\noindent{\bf Remark:}%
\enspace\rm}%
{\myaftermedspace}
\newtheorem{teor}{Theorem}[section]
\newtheorem{defi}[teor]{Definition}
\newtheorem{exam}[teor]{Example}
\newtheorem{lem}[teor]{Lemma}
\newtheorem{cor}[teor]{Corollary}

\newcommand{\beq}{\begin{equation}}
\newcommand{\eeq}{\end{equation}}
\newcommand{\beql}[1]{\begin{equation} \label{#1}}
\newcommand{\eeql}{\end{equation}}
\newcommand{\beqa}{\begin{eqnarray*}}
\newcommand{\eeqa}{\end{eqnarray*}}
\newcommand{\beqal}[1]{\begin{eqnarray} \label{#1}}
\newcommand{\eeqal}{\end{eqnarray}}
\newcommand{\beqan}{\begin{eqnarray}}
\newcommand{\eeqan}{\end{eqnarray}}
\newcommand{\bpf}{\begin{proof}}
\newcommand{\epf}{\end{proof}}

\newcommand{\GF}{{\rm GF}}

\begin{document}
\begin{titlepage}
\title{Optimal codes for correcting a single (wrap-around) burst of
erasures}
\date{\today}
\author{%
Henk D.L.\ Hollmann and Ludo M.G.M.\ Tolhuizen
\thanks{The authors are with Philips Research Laboratories, Prof.\ Holstlaan 4,
5656 AA Eindhoven, The Netherlands;
e-mail:\{henk.d.l.hollmann,ludo.tolhuizen\}@philips.com}%
}
\maketitle
\begin{abstract}
In 2007, Martinian and Trott presented codes for correcting a
burst of erasures with a minimum decoding delay. Their
construction employs $[n,k]$ codes that can correct any burst of
erasures (including wrap-around bursts) of length $n-k$. They
raised the question if such $[n,k]$ codes exist for all integers
$k$ and $n$ with $1\leq k\leq n$ and all fields (in particular,
for the binary field). In this note, we answer this question
affirmatively by giving two recursive constructions and a direct
one.
\end{abstract}
\end{titlepage}
\section{Introduction}
In \cite{MaTr}, Martinian and Trott present codes for correcting a
burst of erasures with a minimum decoding delay. Their
construction  employs $[n,k]$ codes that can correct any burst of
erasures (including wrap-around bursts) of length $n-k$. Examples
of such codes are MDS codes and cyclic codes.
The question is raised in \cite{MaTr} if such $[n,k]$ codes exist
for all integers $k$ and $n$ with $1\leq k\leq n$ and all fields
(in particular, over the binary field). In this note, we answer
this question affirmatively by giving two recursive constructions
and a direct one.

 Throughout this note, all matrices and codes
are over the (fixed but arbitrary) finite field $\mathbb{F}$, and
we restrict ourselves to linear codes. \\ Obviously, a code of
length $n$ can correct a pattern $E$ of erasures if and only if
any codeword can be uniquely recovered from its values in the
$(n-|E|)$ positions outside $E$. As a consequence, if an $[n,k]$
code can correct a pattern $E$ of erasures, then $n-|E|\geq k$,
{\em i.e.},  $|E|\leq n-k$.  We call an $[n,k]$ code {\em optimal}
if it can correct any burst of erasures (including wrap-around
bursts) of length $n-k$.\footnote{A more precise terminology would
be "optimal for the correction of a single (wrap-around) burst of
erasures", but we opted for just "optimal" for notational
convenience.}
 Equivalently, an $[n,k]$ code is
optimal if knowledge of any $k$ (cyclically) consecutive symbols
from a codeword allows one to uniquely recover that codeword, or,
in coding parlance, if each of the $n$ sets of $k$ (cyclically)
consecutive codeword positions forms an information set.
 We call a $k\times n$ matrix {\em good} if any $k$ cyclically
consecutive columns of $G$ are independent. It is easy to see that
a code is optimal if  and only if it has a good generator matrix.
\\ Throughout this note, we
denote with $I_k$ the $k\times k$ identity matrix, and with $X^T$
the transpose of the matrix $X$. \\
\section{A recursive construction of optimal codes}
In this section, we give a recursive construction of good
matrices, and hence of optimal codes. We start with a simple
duality result.
\begin{lem}
 Let $C$ be an $[n,k]$ code, and let $C^{\perp}$ be its
 dual. If $I\subset\{1,\ldots ,n\}$ has size $k$ and is an
 information set for $C$, then $I^{\ast}=\{1,\ldots ,n\}\setminus I$ is an
 information set for $C^{\perp}$.
 \end{lem}
 \begin{proof}
 By contradiction. Suppose that $I^{\ast}$
  is not an information set for $C^{\perp}$. Then there is a non-zero word {\bf x} in $C^{\perp}$
that is zero in the positions indexed by $I^{\ast}$. As ${\bf x}$
is in $C^{\perp}$, for any word {\bf c}$\in$$C$ we have that
\[ 0 = \sum_{i=1}^n x_ic_i = \sum_{i\in I} x_i c_i .  \]
As a consequence, there are  $k$-tuples that do not occur in $I$
in any word of $C$, a contradiction. We conclude that $I^{\ast}$
is an information set for $C^{\perp}$. \end{proof} As a
consequence, we have the following.
\begin{cor}\label{dual}
A linear code is optimal if and only if its dual is optimal.
\end{cor}
Our first theorem shows how to construct a good $k\times (k+n)$
matrix from a good $k\times n$ matrix.
\begin{teor}\label{thm1}
Let $G=(I_k\; P)$ be a good $k\times n$ matrix. Then
$G^{\prime}=(I_k\;I_k\;P)$ is a good $k\times (k+n)$ matrix.
\end{teor}
\begin{proof}
Any $k$ cyclically consecutive columns in $G^{\prime}$ either are
$k$ different unit vectors, or $k$ cyclically consecutive columns
of $G$.
\end{proof}
Our next theorem shows how to construct a good $n\times (2n-k)$
matrix from a good $k\times n$ matrix.
\begin{teor}\label{thm2}
 Let  $G=\left(
I_k \;P\right)$ be a good $k\times n$ matrix. The the following
$n\times (2n-k)$ matrix $G^{\prime}$ is good
\[ G^{\prime} = \left( \matrix{ I_{n-k} & 0 & I_{n-k} \cr
                        0     & I_k & P} \right) . \]
\end{teor}
\begin{proof}  As $G$ is good, Corollary~\ref{dual} implies that the generator
matrix $(-P^T \; I_{n-k})$ of the dual of the code generated by
$G$ is good. By cyclically shifting the columns of this matrix
over
$(n-k)$ positions to the right, we obtain the good matrix $(I_{n-k}\;\;-P^T)$.  \\
Theorem~1 implies that $(I_{n-k}\; I_{n-k}\;\; -P^T)$ is good, and
so the matrix $H=\left( I_{n-k}\;\;-P^T \;I_{n-k}\right)$ obtained
by cyclically shifting the columns of the former matrix over $n$
positions, is good. Clearly, after multiplying the columns of a
good matrix with non-zero field elements, we obtain a good matrix;
as a consequence, $H^{\prime}= \left(\matrix{-I_{n-k}\;\; -P^T \;
I_{n-k}}\right)$ is good. As $H^{\prime}$ is a good full-rank
parity check matrix of the code generated by $G^{\prime}$, this
latter matrix is good.
\end{proof}
{\bf Remark} The construction from Theorem~\ref{thm2} also occurs
in the proof of \cite[Thm.1]{MaTr}. \\ \\
%
The construction from Theorem~\ref{thm1} increases the code length
and fixes its dimension; the construction from Theorem~\ref{thm2}
also increases the code length, but fixes its redundancy. These
constructions can be combined to give a recursive construction of
optimal $[n,k]$ code for all $k$ and $n$. The following definition
is instrumental in making this explicit.
\begin{defi}
For positive integers $r$ and $k$, we recursively define the
$k\times r$ matrix $P_{k,r}$ as follows:
\[ P_{k,r} = \left\{ \begin{array}{cl}
\left( \matrix{ I_r \cr P_{k-r,r}}\right) & \mbox{ if } 1\leq r <
k , \\
  I_k & \mbox{ if } r=k, \\
  \left( \matrix{ I_{k}\; P_{k,r-k}}\right) &  \mbox{ if } r > k.
\end{array} \right. \]
\end{defi}
\begin{teor}\label{thm3}
 For each positive integer $k$, the matrix $I_k$ is good. \\
 For all integers $k$ and $n$ with 1$\leq k <n$, the $k\times n$
 matrix
  $\left( I_k \; P_{k,n-k}\right)$ is good.
\end{teor}
\begin{proof} The first statement is obvious. \\
The second statement will be proved by induction on $k+n$. It is
easily verified that it is true for $k+n=3$. Now assume that the
statement is true for all integers $a,b$ with $1\leq a\leq b$ and
$a+b < k+n$.  We consider three cases.
\\
If $2k < n$, then by induction hypothesis $(I_k \; P_{k,n-2k})$ is
good. By Theorem~\ref{thm1}, $\left( I_k \; I_k \; P_{k,n-2k}
\right) = \left( I_k \; P_{k,n-k}\right)$  is also good.   \\ If
$2k=n$, then $\left( I_k \; P_{n-k} \right) = \left( I_k \;
P_{k,k} \right) = \left( I_k\; I_k\right)$, which obviously is a
good matrix.
 If $k<n$ and $2k>n$, the induction hypothesis implies that
 $(I_{2k-n} P_{2k-n,n-k})$ is a good
 $(2k-n)\times k$ matrix. By Theorem~\ref{thm2},
\[ \left( \matrix{
                 I_{n-k} & 0 & I_{n-k} \cr
                 0   &   I_{2k-n} & P_{2k-n,n-k}} \right) =
                 \left( I_k  P_{k,n-k}\right)   \]
                 is also good.
\end{proof}
\begin{exam}
\begin{rm}
Theorem~$\ref{thm3}$ implies that $(I_{28} P_{28,17})$ is a good
$28\times 45$ matrix. \\
According to the definition, $P_{28,17}=\left( \matrix{ I_{17}\cr
P_{11,17}}\right)$. \\
Again according to the definition, $P_{11,17}=(I_{11}P_{11,6})$. \\
Continuing in this fashion, $P_{11,6}= \left(\matrix{I_6 \cr
P_{5,6}}\right)$. \\
Finally, $P_{5,6}=(I_5 P_{5,1})$, and, as can be readily seen by
induction on $k$,  $P_{k,1}$ is the all-one vector of height $k$.
\\ Putting this altogether, we find that  the following $28\times 45$ matrix $G$ is good:
\[ G = \left( \begin{array}{ccc|cc|ccc}
    I_{6} & 0      & 0     & 0     & 0     & I_{6} & 0    & 0 \\
    0     &  I_5   & 0     & 0     & 0     & 0     & I_5  & 0 \\
    0    & 0     & I_6     & 0     & 0     & 0     & 0   & I_6 \\
    \hline
    0    & 0      & 0      & I_6   & 0     & I_6   & 0   & I_6 \\
    0    & 0      & 0      & 0     & I_5   & 0     & I_5 & P_{5,6}
    \end{array}
    \right) , \]
    where $P_{5,6} = \left( I_5 {\bf 1}\right)$, where ${\bf 1}$
    denotes the all-one column vector.
    \end{rm}
\end{exam}
To close this section, we remark that with an induction argument
it can be shown that for all positive integers $k$ and $r$, we
have $P_{k,r}=P_{r,k}^T$.
\section{Adding one column to a good matrix}
In Theorem~\ref{thm1}, we added $k$ columns to a good $k\times n$
matrix to obtain a good $k\times (k+n)$ matrix. In this section,
we will show that it is always possible to add a single column to
a good $k\times n$ matrix in such a way that the resulting
$k\times (n+1)$ matrix is good; we also show that the in the
binary case, there is a {\em unique} column that can be added. The
desired result is a direct consequence of the following
observation, which may be of independent interest.
\begin{lem}\label{dual2}
Let $\mathbb{F}$ be any field, and let $a_1, a_2, \ldots,
a_{2k-2}$ be a sequence of vectors in $\mathbb{F}^k$ such that
$a_i, a_{i+1}, \ldots, a_{i+k-1}$ are independent over
$\mathbb{F}$ for $i=1, \ldots, k-1$. For $i=1, \ldots, k$, let
$b_i$ be a nonzero vector orthogonal to $a_i, a_{i+1}, \ldots,
a_{i+k-2}$. Then $b_1, \ldots, b_k$ are independent over
$\mathbb{F}$.
\end{lem}
\bpf For $i=1, \ldots ,k$, we define
\[V_i:={\rm span}\{a_i, \ldots, a_{i+k-2}\}. \]
For an interval $[i+1, i+s]:=\{i+1, i+2, \ldots, i+s\}$, with
$0\leq i <i+s\leq k$, we let
\[ V_{[i+1,i+s]}=V_{i+1}\cap \cdots \cap V_{i+s}\]
denote the intersection of $V_{i+1}, \ldots, V_{i+s}$.
Note that by definition
\[V_{[i,i]}=V_i=b_i^\perp.\]
We claim that
\[V_{[i+1,i+s]}={\rm span}\{a_{i+s}, \ldots, a_{i+k-1}\}.\]
This is easily proven by induction on $s$: obviously, the claim is true for
$s=1$; if it holds for all $s'\leq s$, then
\beqa V_{[i+1,i+s+1]}&=&V_{[i+1,i+s]}\cap V_{i+s+1}\\
        &=&{\rm span} \{a_{i+s}, \ldots, a_{i+k-1}\} \cap {\rm
span}\{a_{i+s+1},
\ldots, a_{i+s +k-1}\},
\eeqa
hence $V_{[i+1,i+s]}$ certainly contains $a_{i+s+1}, \ldots, a_{i+k-1}$ and
does not contain $a_{i+s}$, since by assumption $a_{i+s}\notin {\rm
span}\{a_{i+s+1},
\ldots, a_{i+s +k-1}\}$.

So by our claim it follows that
\[ \{0\}= V_{[1,k]}=V_1\cap \cdots V_k = b_1^\perp \cap \cdots b_k^\perp, \]
hence $b_1, \ldots, b_k$ are independent.
\epf

As an immediate consequence, we have the following.
\begin{teor}\label{addcolumn-alt}
Let $M$ be a good $k\times n$ matrix over $\GF(q)$. There are
precisely $(q-1)^k$ vectors $x\in\GF(q)^k$ such that the matrix
$(M x)$ is good.
\end{teor}
\bpf Let $M=(m_0,m_1,\ldots ,m_{n-1})$ have columns $m_0,\ldots
m_{n-1}\in GF(q)^k$. We want to find all vectors $x\in\GF(q)^k$
with the property that the $k$ vectors \beql{vecs} m_{n-i},
\ldots, m_{n-1}, x, m_0,\ldots, m_{k-i-2}\eeql are independent,
for all $i=k-1, k-2, \ldots, 0$. So, for $i=k-1, k-2, \ldots, 0$,
let $b_i$ be a nonzero vector orthogonal to $m_{n-i}, \ldots,
m_{n-1},m_0,\ldots, m_{k-i-2}$; since $M$ is good, the $k-1$
vectors $m_{n-i}, \ldots, m_{n-1},m_0,\ldots, m_{k-i-2}$ are
independent, and hence the vectors in (\ref{vecs}) are independent
if and only if $(x,b_i)=\lambda_i\neq0$. Again since $M$ is good,
the $2k-2$ vectors
\[m_{n-k+1}, \ldots, m_{n-1}, m_0, \ldots, m_{k-2}\]
satisfy the conditions in Lemma~\ref{dual2}, hence the vectors
$b_0, \ldots, b_{k-1}$ are independent. So for each choice of
$\lambda=(\lambda_0, \ldots, \lambda_{k-1})$ with $\lambda_i\neq0$
for each $i$, there is a unique vector $x$ for which
$(x,b_i)=\lambda_i$, and these vectors $x$ are precisely the ones
for which $(M x)$ is good. \epf

\section{Explicit construction of good matrices}
By starting with the $k\times k$ identity matrix, and repeatedly
applying Theorem~\ref{addcolumn-alt}, we find that for each field
$\mathbb{F}$ and all positive integers $k$ and $n$ with $n\geq k$,
there exists a $k\times n$ matrix
$G$ such that \\
(1) the $k$ leftmost columns of $G$ form the $k\times k$
identity matrix, and \\
(2) for each $j$, $k\leq j\leq n$, the $j$ leftmost columns of
$G$ form a good $k\times j$ matrix. \\
Note that Theorem~\ref{addcolumn-alt} implies that for the binary
field, these matrices are unique. It turned out that they have a
simple recursive structure, which inspired our general
construction.

In this section, we give, for all positive integers $k$ and $n$
with $k\leq n$, an explicit construction of $k\times n$ matrices
over $\mathbb{Z}_p$, the field of integers modulo $p$, that
satisfy the above properties (1) and (2). Note that such matrices
also satisfy (1) and (2) for extension fields of $\mathbb{Z}_p$.

We start with describing the result for $p=2$. Let $M_1$ be the
matrix
\begin{equation} \label{M1}
 M_1 = \left( \matrix{ 1 & 0 \cr 1 & 1} \right) ,
 \end{equation} and for
$m\geq 1$, let $M_{m+1}$ be the given as
\begin{equation} \label{binrecur}
 M_{m+1} = \left( \matrix{M_m & 0 \cr M_m & M_m} \right) .
 \end{equation}
Clearly, $M_m$ is a binary 2$^m\times 2^m$ matrix. The relevance
of the matrix $M_m$ to our problem is explained in the following
theorem.
\begin{teor}\label{binresult}
Let $k$ and $r$ be two positive integers, and let $m$ be the
smallest integer such that $2^m\geq k$ and $2^m\geq r$. Let $Q$ be
the $k\times r$ matrix residing in the lower left corner of
$M_{m}$. Then for each integer $j$ for which $k\leq j\leq k+r$,
the $j$ leftmost columns of the matrix $(I_k \; Q)$ form a good
binary $k\times j$ matrix.
\end{teor}
Theorem~\ref{binresult} is a consequence from our results for the
general case in the remainder of this section. \\ \\
We now define the matrices that are relevant for constructing good
matrices over $\mathbb{Z}_p$.
\begin{defi}\label{expldef} Let $p$ be a prime number, and let $k,r$ be positive integers.
 Let $m$ be the smallest integer such that $p^m\geq r$ and $p^m\geq k$.
  The $k\times r$ matrix $Q_{k,r}$ is defined as
\[ Q_{k,r}(i,j) = {p^m-k+i-1 \choose j-1}  \;\;\mbox{ for }
     1\leq i\leq k, 1\leq j\leq r . \]
\end{defi}

In Theorem~\ref{exppf} we will show that the matrix $(I_k\;
Q_{k,r})$ is good over $\mathbb{Z}_p$. But first, we derive a
recursive property of the $Q$-matrices. To this aim, we need some
well-known results on binomial coefficients modulo $p$.
\begin{lem}\label{binop}
Let $p$ be a prime number, and let $m$ be a positive integer. For
any integer $i$ with $1\leq i\leq p^m-1$, we have that ${p^m
\choose i}\equiv 0 \bmod p$.
\end{lem}
\begin{proof}
The following proof was pointed out to us by our colleague Ronald
Rietman. \\
Let $1\leq i\leq p^m-1$. We have that
\[ {p^m\choose i} = \frac{p^m{p^m-1\choose i-1}}{i}. \]
In the above representation of ${p^m\choose i}$, the nominator
contains at least $m$ factors $p$, while the denominator contains
at most $m-1$ factors $p$.
\end{proof}
\begin{lem}\label{Lucas}
Let $p$ be a prime number, and let $m$ be a positive integer.
Moreover, let $i,j,k,\ell$ be integers such that $0\leq i,k\leq
p-1$ and $0\leq j,\ell\leq p^m-1$. Then we have that
\[ {ip^m + j \choose kp^m+\ell}\equiv {i\choose k}{j\choose \ell}
\bmod p. \] \end{lem}
\begin{proof}
This is a direct consequence of Lucas' theorem (see for example
\cite[Thm.\ 13.3.3]{Blahut}). We give a short direct proof.
Clearly, ${ip^m+j \choose kp^m+\ell}$ is the coefficient of
$z^{kp^m+\ell}$ in $(1+z)^{ip^m+j}$. Now we note that
\[ (1+z)^{ip^m+j} = (1+z)^{ip^m}(1+z)^j = \left[
(1+z)^{p^m}\right]^i (1+z)^{j}. \] It follows from
Lemma~\ref{binop} that $(1+z)^{p^m}\equiv 1+z^{p^m} \bmod p$, and
so
\[ (1+z)^{ip^m+j} \equiv (1+z^{p^m})^i (1+z)^j \bmod p . \]
Hence, modulo $p$, the coefficient of $z^{kp^m+\ell}$ in
$(1+z)^{ip^m+j}$ equals ${i\choose k}{j\choose \ell}$.
\end{proof}
\begin{cor}\label{recur}
Let $p$ be a prime, and let $m$ be a positive integer. Let
$a,b,c,d$ be integers such that $0\leq a,c\leq p-1$ and $1\leq
b,d\leq p^m$. Then we have
\[ Q_{p^{m+1},p^{m+1}}(ap^m+b,cp^m+d) \equiv {a\choose c}
Q_{p^m,p^m}(b,d) \bmod p . \]
\end{cor}
\begin{proof}
According to the definition of $Q_{p^{m+1},p^{m+1}}$, we have that
\[ Q_{p^{m+1},p^{m+1}}(ap^m+b,cp^m+d) =
   {ap^m+b-1 \choose cp^m+d-1}, \mbox{ and } Q_{p^m,p^m}(b,d) ={b-1 \choose d-1}. \]
The corollary is now obtained by application of Lemma~\ref{Lucas}.
\end{proof}
In words, Theorem~\ref{recur} states that $Q_{p^{m+1},p^{m+1}}$
can be considered as a $p\times p$ block matrix, for which each
block is a multiple of $Q_{p^m,p^m}$. For example, for $p=3$, we
obtain
\[ Q_{3^{m+1},3^{m+1}}=
    \left( \matrix{ {0\choose 0} & {0\choose 1} & {0\choose 2} \cr
                    {1\choose 0} & {1\choose 1} & {1\choose 2} \cr
                     {2\choose 0}& {2\choose 1} & {2\choose
                     2}} \right)
                     \times Q_{3^m,3^m} =
     \left(      \matrix{ Q_{3^m,3^m} & 0  & 0 \cr
                    Q_{3^m,3^m} & Q_{3^m,3^m} & 0 \cr
                    Q_{3^m,3^m} & 2Q_{3^m,3^m} & Q_{3^m,3^m}}
     \right) . \]
     For $p=2$, we obtain the relation in (\ref{binrecur}).

Taking $a=p-1$ and $c=0$ in Theorem~\ref{recur}, we see that over
$\mathbb{Z}_p$, the $p^m\times p^m$ block in the lower left hand
corner of $Q_{p^{m+1},p^{m+1}}$ equals $Q_{p^m,p^m}$.
Definition~\ref{expldef}  implies $Q_{k,r}$ is the $k\times r$
matrix residing in the lower left hand corner of $Q_{p^m,p^m}$,
where $m$ is the smallest integer that such that $p^m\geq k$ and
$p^m\geq r$. The above observations imply that whenever
$k^{\prime}\geq k$ and $r^{\prime}\geq r$, then over
$\mathbb{Z}_p$, the matrix $Q_{k,r}$ is the $k\times r$ submatrix
in the lower left hand corner of $Q_{k^{\prime},r^{\prime}}$. In
particular, $Q_{k,r+1}$ can be obtained by adding a column to
$Q_{k,r}$.

We now state and prove results on the invertibility in
$\mathbb{Z}_p$ of certain submatrices of $Q_{k,r}$, that will be
used to prove our main result in Theorem~\ref{exppf}.

\begin{lem}\label{intinv}
Let $n\geq 0$ and $b\geq 1$. The $b\times b$ matrix $V_b$ with
$V_b(i,j)={n+i-1\choose j-1}$ for $1\leq i,j\leq b$ has an integer
inverse.
\end{lem}
\begin{proof} By induction on $b$.
For $b=1$, this is obvious. \\
Next, let $b\geq 2$. Let $S$ be the $b\times b$ matrix with
\[ S(i,j) = \left\{ \begin{array} {ll}
     1 & \mbox{ if } i=j, \\
     -1 & \mbox{ if } i\geq 2 \mbox{ and } i=j+1, \\
     0  & \mbox{ otherwise.} \end{array} \right. \]
The matrix $S$ has an integer inverse: it is easy to check that
$S^{-1}(i,j)=1$ if $i\geq j$, and 0 otherwise. We have that
\[ (SV_b)(1,j)=V_b(1,j)={n \choose j-1} , \mbox{ and } \]
\[(SV_b)(i,j) = V_b(i,j)-V_b(i-1,j)={n+i-1\choose j-1} - {n+i-2\choose j-1} =
   {n+i-2 \choose j-2} \mbox{ for } 2\leq j \leq b. \]
In other words,  $SV_b$ is of the form
\[ SV_b = \left( \matrix{ 1 & A \cr 0 & V_{b-1}}\right) . \]
By induction hypothesis, $V_{b-1}$ has an integer inverse, and so
$V_bS$ has an integer inverse (namely the matrix $\left( \matrix{
1 & -A V_{b-1}^{-1} \cr 0  & V_{b-1}^{-1}}\right)$). As $S$ has an
integer inverse, we conclude that $V_b$ has an integer inverse.
\end{proof}
\begin{lem}\label{modpinv2}
Let $p$ be a prime number, and let $a\geq 0$ and $b\geq 1$ be
integers such that $a+b\leq p^m$.  The $b\times b$ matrix $W_b$
with $W_b(i,j) = {p^m-1+i-b \choose a+j-1}$ for 1$\leq i,j\leq b$
is invertible over $\mathbb{Z}_p$.
\end{lem}
\begin{proof}
Similarly to the proof of Lemma~\ref{intinv}, we apply
induction on $b$. \\
For $b=1$, the we have the 1x1 matrix with entry  ${p^m-1 \choose
a}$. By induction on $i$, using that ${p^m-1 \choose i}
={p^m\choose i}-{p^m-1 \choose i-1}$ and employing
Lemma~\ref{binop}, we readily find that ${p^m-1 \choose i} \equiv
(-1)^i \bmod p \mbox{ for } 0\leq i\leq p^m-1$. As a consequence,
the lemma is true for
$b=1$. \\
 Now let $b\geq 2$. We define the $b\times b$ matrix $T$ by
\[ T(i,j) = \left\{ \begin{array}{ll} 1 & \mbox{ if } i=j \cr
                            1 & \mbox{ if } j\geq 2 \mbox{ and }
                            i=j-1 \cr
                            0 & \mbox{ otherwise}
                            \end{array} \right. \]
It is easy to check $T$ has an integer inverse, and that
$T^{-1}(i,j)=(-1)^{i-j}$ if $i\leq j$ and 0 otherwise. In order to
show that $W_b$ is invertible in $\mathbb{Z}_p$, it is thus
sufficient to show that $W_bT$ is invertible in $\mathbb{Z}_p$. By
direct computation, we have that $(W_bT)(i,1)=W_b(i,1)$, and
\[ (W_bT)(i,j) = W_b(i,j)+W_b(i,j-1) =
  {p^m-1+i-b \choose a+j-1} + {p^m-1+i-b \choose a+j-2} =
   {p^m+i-b \choose a+j-1} . \]
In particular, $(W_bT)(b,1)={p^m-1 \choose a}\equiv (-1)^a \bmod
p$, and for $2\leq j\leq b$, we have that $(W_bT)(b,j)={p^m
\choose a+j-1}\equiv 0\bmod p$. We thus have that
\[ W_bT \equiv \left(\matrix{A & W_{b-1} \cr (-1)^a & 0} \right) \bmod p . \]
As $W_{b-1}$ is invertible over $\mathbb{Z}_p$, the matrix $W_bT$
(and hence the matrix $W_b$) is invertible over $\mathbb{Z}_p$.
                            \end{proof}

{\bf Remark} The matrix in Lemma~\ref{modpinv2} need not have an
integer inverse. For example, take $p=2, m=2, a=1$ and $b=2$. The
matrix $W_2$ equals
\[ \left( \matrix{{2\choose 1} & {3\choose 1} \cr {2\choose 2} &
{3\choose 2}}\right) = \left( \matrix{ 2 & 3 \cr 1 & 3}\right),
\] and so $W_2^{-1}= \left( \matrix{1 & -1  \cr -\frac{1}{3} &
\frac{2}{3}}\right)$. Note that modulo 2, $W_2$ equals $\left(
\matrix{0 & 1 \cr 1 & 1}\right)$, confirming that $W_2$ does have
an inverse in the integers modulo $p=2$.

\vspace{0.5cm}
We are now in a position to prove the main result of this section.
\begin{teor}\label{exppf}
Let $k$ and $r$ be positive integers. For  $j=k,k+1,\ldots k+r$,
the matrix consisting of the $j$ leftmost columns of the matrix
$(I_k \;Q_{k,r})$ is good over $\mathbb{Z}_p$.
\end{teor}
\begin{proof}
We denote the matrix $(I_k \; Q_{k,r})$ by $G$, and the $i$-th
column of $G$ by {\bf g}$_i$. Let $k\leq j\leq k+r$. To show that
the matrix consisting of the columns 1,2,\ldots, $j$ of $G$ is
good, we show that for $1\leq i\leq j$, the vectors ${\bf
g}_i,{\bf g}_{i+1},\ldots ,{\bf g}_{i+k-1}$ are independent over
$\mathbb{Z}_p$, where the indices are counted modulo $j$. This is
obvious if $j=k$ and if $i=1$, so we assume that $j\geq k+1$ and
$i\geq 2$. We distinguish between two cases.
\\
(1) $2\leq i\leq k$. \\
The vectors to consider are ${\bf e}_i,\ldots ,{\bf e}_k, {\bf
g}_{k+1},\ldots ,{\bf g}_{i+k-1}$ (if $i+k-1\leq j$), or ${\bf
e}_i,\ldots ,{\bf e}_{k},\\ {\bf g}_{k+1},\ldots ,{\bf g}_{j},{\bf
e}_1,\ldots ,{\bf e}_{k-j+i-1}$ (if $i+k-1\geq j+1$). We define
$b$:=min($i-1,j-k$). The vectors under consideration are
independent if the $b\times b$ matrix consisting of the $b$
leftmost columns of $Q_{k,r}$, restricted to rows
$i-b,i-b+1,\ldots ,i=1$, is invertible in $\mathbb{Z}_p$.
This follows from Lemma~\ref{intinv}. \\
(
2) $i\geq k+1$. \\
The vectors to consider are ${\bf g}_{i},\ldots ,{\bf g}_{i+k-1}$
(if $i+k-1\leq j$), or ${\bf g}_{i},\ldots ,{\bf g}_{j},{\bf
e}_1,\ldots ,{\bf e}_{k-j+i-1}$ (if $i+k-1\geq j+1$). We define
$b$:=min($k,j-i+1$). The vectors under consideration are
independent if  the $b\times b$ matrix consisting of the $b$
bottom entries of the columns $i-k+1,i-k+2,\ldots,i-k+b$ of
$Q_{k,r}$ is invertible in $\mathbb{Z}_p$. This follows from
Lemma~\ref{modpinv2}.
\end{proof}

      \end{document}